\documentclass[journal,twoside]{IEEEtran}

\usepackage[T1]{fontenc}
\usepackage{float}
\usepackage[utf8]{inputenc}

\usepackage{amsmath,amssymb,amsfonts,mathtools}
\usepackage{amsthm}
\usepackage{bm}
\usepackage{placeins}

\usepackage{cite}
\usepackage{url}

\usepackage{graphicx}
\usepackage{booktabs}
\usepackage{array}
\usepackage{multirow}

\usepackage{algorithm}
\usepackage{algpseudocode}

\usepackage{xcolor}

\usepackage{cuted}
\usepackage{capt-of}
\usepackage{dblfloatfix}

\usepackage{enumitem}

\setlist{
    topsep=1pt,
    itemsep=0pt,
    parsep=0pt,
    partopsep=0pt
}

\usepackage{float}

\usepackage{titlesec}

\titlespacing*{\section}
{0pt}
{0.5ex}
{0.2ex}

\titlespacing*{\subsection}
{0pt}
{0.4ex}
{0.15ex}

\titlespacing*{\subsubsection}
{0pt}
{0.3ex}
{0.1ex}

\allowdisplaybreaks
\AtBeginDocument{%
    \setlength{\abovedisplayskip}{2pt}%
    \setlength{\belowdisplayskip}{2pt}%
    \setlength{\abovedisplayshortskip}{1pt}%
    \setlength{\belowdisplayshortskip}{1pt}%
    \setlength{\jot}{1pt}%
}

\newtheorem{lemma}{Lemma}
\newtheorem{proposition}{Proposition}
\newtheorem{theorem}{Theorem}

\newtheorem{remark}{Remark}

\newcommand{\bbR}{\mathbb{R}}

\newcommand{\bbE}{\mathbb{E}}

\newcommand{\R}{\mathbb{R}}

\newcommand{\calB}{\mathcal{B}}

\newcommand{\calD}{\mathcal{D}}

\newcommand{\calL}{\mathcal{L}}
\newcommand{\calM}{\mathcal{M}}

\newcommand{\calS}{\mathcal{S}}

\newcommand{\calX}{\mathcal{X}}

\newcommand{\mat}[1]{\mathbf{#1}}
\newcommand{\vect}[1]{\bm{#1}}

\newcommand{\tran}{^{\mathsf T}}

\newcommand{\xv}{\vect{x}}
\newcommand{\uv}{\vect{u}}
\newcommand{\vv}{\vect{v}}
\newcommand{\wv}{\vect{w}}
\newcommand{\zv}{\vect{z}}
\newcommand{\sv}{\vect{s}}

\newcommand{\pv}{\mathbf{p}}

\newcommand{\xiv}{\vect{\xi}}

\newcommand{\thetav}{\vect{\theta}}
\newcommand{\lambdav}{\vect{\lambda}}

\newcommand{\bv}{\vect{b}}

\newcommand{\hv}{\vect{h}}
\newcommand{\cv}{\vect{c}}

\newcommand{\Am}{\mat{A}}
\newcommand{\Bm}{\mat{B}}
\newcommand{\Cm}{\mat{C}}
\newcommand{\Dm}{\mat{D}}
\newcommand{\Em}{\mat{E}}
\newcommand{\Fm}{\mat{F}}

\newcommand{\Hm}{\mat{H}}
\newcommand{\Id}{\mat{I}}

\newcommand{\Lm}{\mat{L}}
\newcommand{\Mm}{\mat{M}}
\newcommand{\Pm}{\mat{P}}
\newcommand{\Qm}{\mat{Q}}
\newcommand{\Rm}{\mat{R}}
\newcommand{\Sm}{\mat{S}}

\newcommand{\Wm}{\mat{W}}

\newcommand{\Zm}{\mat{Z}}

\newcommand{\Lam}{\boldsymbol{\Lambda}}

\DeclareMathOperator{\diagop}{diag}

\DeclareMathOperator{\colop}{col}

\DeclareMathOperator*{\argminop}{arg\,min}

\newcommand{\lmax}{\lambda_{\max}}

\newcommand{\norm}[1]{\left\lVert #1\right\rVert}

\newcommand{\setc}[1]{\left\{#1\right\}}

\usepackage{hyperref}

\hypersetup{
    hidelinks,
    pdfborder={0 0 0}
}
\usepackage{tikz}
\usepackage{hyperref}
\usetikzlibrary{svg.path}

\usepackage{orcidlink}

\usepackage{graphicx}
\usepackage{booktabs}
\usepackage{placeins}
\usepackage{amsmath,amssymb,bm}

\hypersetup{
    colorlinks=true,
    linkcolor=blue,
    citecolor=blue,
    urlcolor=blue,
    filecolor=blue
}

\newcommand{\safeincludegraphics}[2][]{%
\IfFileExists{#2}{%
    \includegraphics[#1]{#2}%
}{%
    \fbox{%
    \begin{minipage}[c][0.18\textheight][c]{0.92\linewidth}
    \centering
    Missing figure: \texttt{\detokenize{#2}}\\[1mm]
    Replace this box by the experimental plot.
    \end{minipage}%
    }%
}%
}

\title{Stability-Aware Imitation Learning from Model Predictive Control for Autonomous Vehicle Lateral Control: Exact Q-Loss and a Novel Training Procedure}

\author{
Tien Dat Vu,
\IEEEmembership{Student Member, IEEE},
Minh Quan Nguyen,
Anh Tuan Vu,
Thanh Tung Nguyen,
and Minh Doan, \IEEEmembership{Member, IEEE}
\thanks{T. D. Vu and M. Doan are with the Faculty of Mechanical Engineering,
Ho Chi Minh City University of Technology (HCMUT), Vietnam National University
Ho Chi Minh City (VNU-HCM), Ho Chi Minh City, Vietnam.
E-mail: \{dat.vuv, minh.doan\}@hcmut.edu.vn.}
\thanks{M. Q. Nguyen, A. T. Vu, and T. T. Nguyen are with the Department of Automotive Engineering,
Hanoi University of Science and Technology (HUST), Hanoi, Vietnam.}
\thanks{Corresponding author: Minh Doan
(e-mail: minh.doan@hcmut.edu.vn).}
}

\begin{document}

\raggedbottom

\maketitle

\begin{abstract}
This paper develops a certified imitation-learning framework for approximating model predictive control (MPC) policies with feedforward neural controllers and validates it on autonomous-vehicle lateral control. An exact finite-horizon \(Q\)-loss is constructed by fixing the learner's first steering action in the expert MPC problem and re-optimizing the remaining horizon, thereby measuring its downstream optimal-control consequence rather than only pointwise action mismatch. The neural policy is represented as a linear fractional transformation (LFT) interconnection with activation nonlinearities described by sector integral quadratic constraints (IQCs). Combined with a quadratic Lyapunov condition, this representation yields a differentiable certification margin based on the largest eigenvalue of the Lyapunov--IQC matrix. The margin is enforced during training through a logarithmic barrier, while certified Dataset Aggregation (DAgger) and safe projection keep data-aggregation rollouts within the certified policy set. Experiments on a CAD-referenced autonomous-vehicle platform with AprilTag localization and real-time steering demonstrate the resulting closed-loop performance.
\end{abstract}


\section{Introduction}

\IEEEPARstart{M}{odel} predictive control (MPC) is attractive for
autonomous-vehicle lateral control because prediction, actuator and
steering-rate constraints, terminal requirements, and finite-horizon
performance can be handled within a single constrained optimization
problem. Its main limitation is computational: repeated online
optimization can be expensive for embedded deployment. Neural imitation
offers inexpensive real-time evaluation, but a policy that closely
matches expert steering commands is not necessarily safe or
closed-loop reliable. In particular, two steering actions with similar
instantaneous imitation errors may have very different finite-horizon
consequences near high-curvature segments or active constraints, while
learner deployment can induce states absent from the original expert
distribution.

This paper addresses these two issues jointly through a certified
imitation-learning framework for replacing an online MPC controller by
a feedforward neural policy. Its central idea is to transfer not only
the expert action but also the finite-horizon value structure that
produced that action. For each learner steering command, the action is
inserted as the first MPC move and the remaining horizon is
re-optimized, yielding an exact-Q quantity that measures its downstream
optimal-control consequence~\cite{Ghezzi2023ExactQ}. This task-aware 
learning objective is coupled with a Lyapunov--integral quadratic 
constraint (IQC) certificate so that closed-loop certification is 
enforced during learning rather than checked only after training. 
Certified Dataset Aggregation (DAgger) further restricts 
learner-induced rollouts to certified policies, while safe projection 
prevents accepted parameter updates from leaving the certified set.

MPC and its learned approximations have been studied extensively.
Classical MPC provides systematic treatment of constrained prediction
and control \cite{Mayne2000,Rawlings2017,Camacho2013,Garcia1989},
while explicit and neural approximations reduce online computational
cost and can provide performance guarantees
\cite{Hertneck2018,Karg2020}. Learning from demonstrations and
imitation learning transfer expert behavior to neural controllers
\cite{Pomerleau1989,Argall2009,Hussein2017,Ghezzi2023ExactQ}, 
and DAgger-type methods reduce distribution shift by querying the 
expert at learner-induced states \cite{Ross2011,Ho2016}; 
safety-aware variants such as DropoutDAgger additionally account for 
uncertain learner rollouts \cite{Menda2017}. In parallel, 
differentiable optimization has been embedded directly into learning 
through OptNet \cite{Amos2017}, differentiable MPC \cite{Amos2018}, 
and differentiable convex optimization layers \cite{Agrawal2019}.
However, these approaches do not by themselves provide a closed-loop
Lyapunov certificate for the deployed neural policy.

Neural-policy certification has therefore developed largely as a
separate line of work. Integral quadratic constraints provide a
systematic representation of nonlinear feedback operators
\cite{Megretski1997}, linear matrix inequalities support Lyapunov and
robust-control analysis \cite{Boyd1994}, IQC methods have been extended
to optimization-related nonlinear mappings \cite{Lessard2016}, and
Lipschitz-based bounds enable tractable neural-network analysis
\cite{Fazlyab2019,Virmaux2018}. These tools can certify a fixed neural
controller, but they do not directly provide a task-aware mechanism for
learning an MPC policy while maintaining certification as the policy
and its training distribution evolve.

The resulting gap is therefore structural. Existing neural MPC
approximations mainly learn to reproduce expert actions without
explicitly measuring the finite-horizon consequence of a learner
action, whereas existing certification methods mainly verify a fixed
policy after it has been obtained. What is missing is a framework in
which the MPC value structure and the stability certificate both shape
the imitation process itself. The proposed method closes this gap by
combining exact-Q supervision with Lyapunov--IQC-constrained learning
and certified data aggregation.

The main contributions are as follows. First, an exact MPC-induced
finite-horizon Q-loss~\cite{Ghezzi2023ExactQ} evaluates each learner 
steering action through its downstream optimal-control consequence 
rather than pointwise action error alone. Second, a path-referenced 
lateral-control formulation connects bicycle dynamics, road curvature, 
lateral and heading errors, yaw-rate feedforward, steering constraints, 
and the finite-horizon expert MPC problem. Third, the neural controller 
is represented through an LFT-like interconnection with 
sector-IQC-described activation nonlinearities, yielding a quadratic 
Lyapunov--IQC certificate and an explicit largest-eigenvalue 
certification margin. Fourth, this margin is embedded into policy 
learning through an interior-barrier objective and a safe-projection 
mechanism. Fifth, certified DAgger restricts learner-induced data 
collection to certified neural policies. Finally, the complete framework 
is validated on a physical autonomous-vehicle platform using a 
CAD-referenced path, onboard lane perception, AprilTag-based external 
localization for performance evaluation, and real-time steering 
actuation, with assessment based on steering imitation, exact-Q 
degradation, and measured closed-loop tracking performance.

\textbf{Notation:} Let $\bbR^n$ denote the $n$-dimensional Euclidean
space. For $\Pm=\Pm\tran$, the relations $\Pm\succ0$ and
$\Pm\succeq0$ denote positive definiteness and positive
semidefiniteness, respectively. For $\zv\in\bbR^n$,
$\norm{\zv}_{\Pm}^{2}:=\zv\tran\Pm\zv$. The identity matrix is denoted
by $\Id$. For a symmetric matrix $\Mm$, $\lmax(\Mm)$ denotes its
largest eigenvalue. The symbol
\(
\mathbf{1}
\)
denotes the vector of all ones of compatible dimension and the symbol
\(
\lesssim
\)
denotes an asymptotic inequality valid up to higher-order terms.

\section{Vehicle Lateral Dynamics and Error-State Model}

\subsection{Continuous-time bicycle model}


The experimental vehicle is modeled by the standard small-angle
bicycle lateral dynamics. Let $v_y$ be the lateral velocity at the
center of gravity, $r=\dot{\psi}$ the yaw rate, $V_x$ the longitudinal
speed, $m$ the vehicle mass, $I_z$ the yaw moment of inertia,
$L_f$ and $L_r$ the distances from the center of gravity to the
front and rear axles, and $\delta_f$ the front steering angle.
The vehicle geometry and the adopted sign conventions are illustrated
in Fig.~\ref{fig:bicycle_model}.

\begin{figure}[H]
    \centering
    \includegraphics[width=\columnwidth]{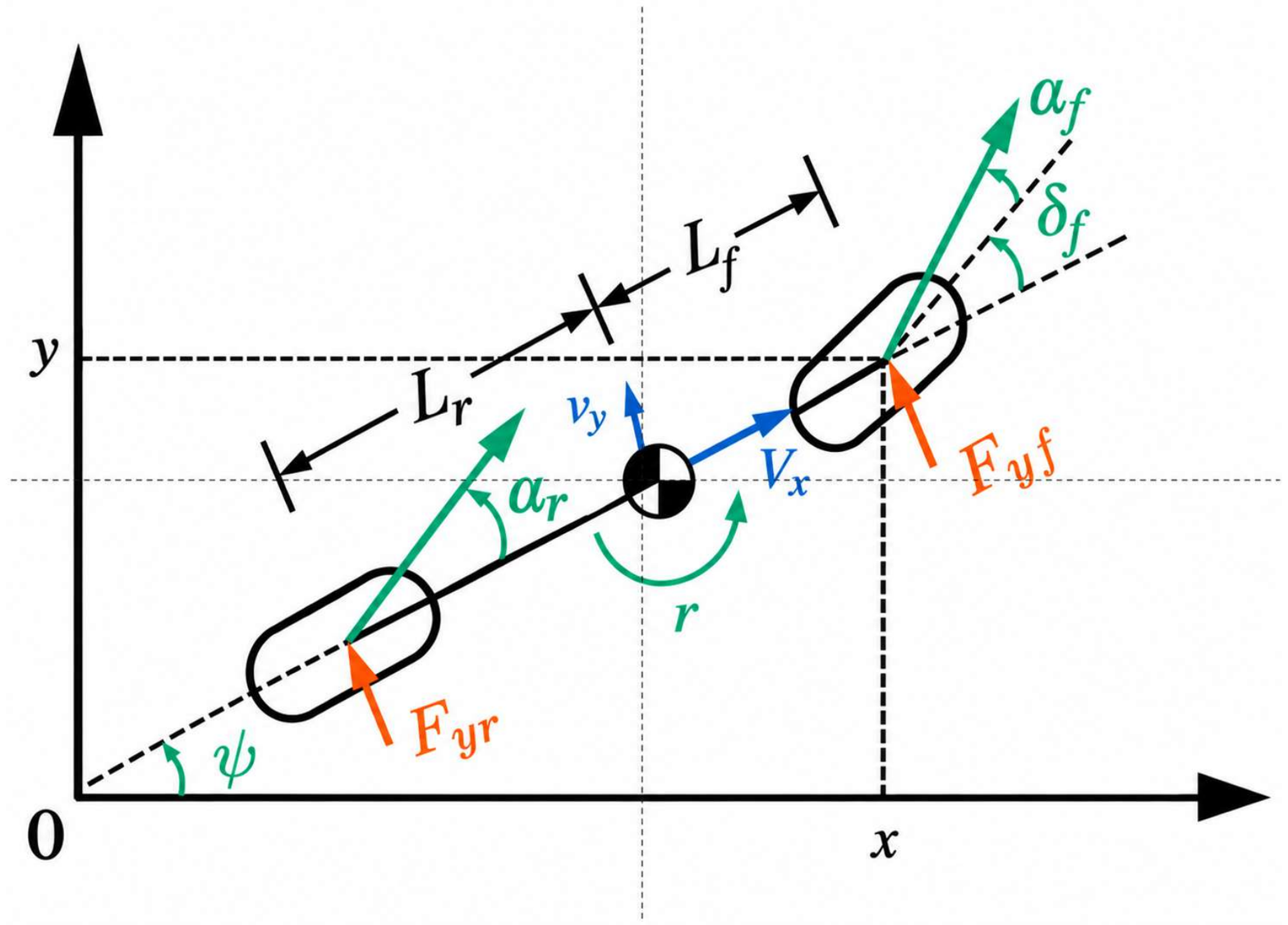}
    \caption{Bicycle model geometry and the adopted sign conventions.}
    \label{fig:bicycle_model}
\end{figure}

\begin{align}
m(\dot v_y+V_xr)
&=
F_{yf}+F_{yr},
\label{eq:lateral_force_balance}\\
I_z\dot r
&=
L_fF_{yf}-L_rF_{yr}.
\label{eq:yaw_moment_balance}
\end{align}

For the operating regime considered in this work, the slip angles
remain sufficiently small so that the tire forces can be approximated
by linear cornering-stiffness relations:
\begin{equation}
F_{yf}
=
2C_{\alpha f}
\left(
\delta_f-\frac{v_y+L_fr}{V_x}
\right).
\label{eq:front_tire_force}
\end{equation}
\begin{equation}
F_{yr}
=
2C_{\alpha r}
\left(
-\frac{v_y-L_rr}{V_x}
\right).
\label{eq:rear_tire_force}
\end{equation}

Substituting \eqref{eq:front_tire_force}--\eqref{eq:rear_tire_force}
into \eqref{eq:lateral_force_balance}--\eqref{eq:yaw_moment_balance}
gives
\begin{equation}
\begin{aligned}
\dot v_y
={}&
-\frac{2C_{\alpha f}+2C_{\alpha r}}
      {mV_x}\,v_y
\\
&-
\left(
V_x+
\frac{2C_{\alpha f}L_f-2C_{\alpha r}L_r}
     {mV_x}
\right)r
\\
&+
\frac{2C_{\alpha f}}{m}\,\delta_f .
\end{aligned}
\label{eq:vy_model}
\end{equation}
\begin{equation}
\begin{aligned}
\dot r
={}&
-\frac{2C_{\alpha f}L_f-2C_{\alpha r}L_r}
      {I_zV_x}\,v_y
\\
&-
\frac{2C_{\alpha f}L_f^2+2C_{\alpha r}L_r^2}
     {I_zV_x}\,r
\\
&+
\frac{2C_{\alpha f}L_f}{I_z}\,\delta_f .
\end{aligned}
\label{eq:r_model}
\end{equation}

These two equations reveal the essential physical structure. The
steering angle affects the vehicle through the front lateral tire
force; the lateral velocity and yaw rate are coupled through both tire
slip and the kinematic term \(V_xr\); and the path curvature appears
later through the desired yaw motion. This is the layer at which a
controller must be physically meaningful: it must not only reduce an
abstract tracking error but also generate steering commands compatible
with the yaw inertia, tire stiffness, and longitudinal speed.

\subsection{Path-referenced error coordinates}

Let the reference path be a smooth planar curve
\begin{equation}
\bm{p}_{\mathrm{ref}}(s)
=
\begin{bmatrix}
x_{\mathrm{ref}}(s)\\
y_{\mathrm{ref}}(s)
\end{bmatrix},
\end{equation}
where \(s\) is the path coordinate. Its tangent angle and curvature
are
\begin{equation}
\psi_{\mathrm{ref}}(s)
=
\mathrm{atan2}
\left(
\frac{dy_{\mathrm{ref}}}{ds},
\frac{dx_{\mathrm{ref}}}{ds}
\right).
\label{eq:path_heading}
\end{equation}
\begin{equation}
\kappa(s)
=
\frac{
x'_{\mathrm{ref}}(s)y''_{\mathrm{ref}}(s)
-
y'_{\mathrm{ref}}(s)x''_{\mathrm{ref}}(s)
}{
\left[
(x'_{\mathrm{ref}}(s))^2+
(y'_{\mathrm{ref}}(s))^2
\right]^{3/2}
}.
\label{eq:path_curvature}
\end{equation}

Given the current vehicle position
\(\bm{p}=[x\ y]\tran\), the nearest path coordinate is computed as
\begin{equation}
s^\star
=
\argminop_s
\norm{
\bm{p}-\bm{p}_{\mathrm{ref}}(s)
}^{2}.
\end{equation}

The lateral and heading errors are then defined with respect to the
path tangent:
\begin{equation}
\begin{aligned}
e_y
={}&
-\sin\psi_{\mathrm{ref}}(s^\star)
\left(
x-x_{\mathrm{ref}}(s^\star)
\right)
\\
&+
\cos\psi_{\mathrm{ref}}(s^\star)
\left(
y-y_{\mathrm{ref}}(s^\star)
\right).
\end{aligned}
\label{eq:lateral_error_def}
\end{equation}
\begin{equation}
e_\psi
=
\psi-\psi_{\mathrm{ref}}(s^\star).
\label{eq:heading_error_def}
\end{equation}

Under the small-angle approximation, the error kinematics are
\begin{align}
\dot e_y
&=
v_y+V_xe_\psi,
\label{eq:ey_dot}\\
\dot e_\psi
&=
r-\dot\psi_{\mathrm{des}}.
\label{eq:epsi_dot}
\end{align}
The desired yaw rate is generated by the curvature of the reference
path,
\begin{equation}
\dot\psi_{\mathrm{des}}
=
V_x\kappa(s^\star).
\label{eq:desired_yaw_rate}
\end{equation}

The state used for lateral control is
\begin{equation}
\begin{aligned}
\xv
&=
\colop
\left(
e_y,\dot e_y,e_\psi,\dot e_\psi
\right).
\end{aligned}
\label{eq:error_state}
\end{equation}

Thus, stabilizing \(\xv\) to the origin means that the vehicle
converges simultaneously in lateral displacement, lateral-error
velocity, heading misalignment, and heading-error rate.

\subsection{Linear error-state realization}


Combining \eqref{eq:vy_model}--\eqref{eq:r_model} with
\eqref{eq:ey_dot}--\eqref{eq:epsi_dot}, and defining
\(
d(t):=\dot\psi_{\mathrm{des}}(t)
\),
one obtains the continuous-time path-referenced realization
\begin{equation}
\dot\xv(t)
=
\Am_c\xv(t)
+
\Bm_c\delta_f(t)
+
\Em_c d(t).
\label{eq:continuous_error_model}
\end{equation}
For compactness, define the vehicle-dependent coefficients
\begin{align}
a_{v1}
&:=
-\frac{2C_{\alpha f}+2C_{\alpha r}}
       {mV_x},
&
a_{v2}
&:=
\frac{2C_{\alpha f}+2C_{\alpha r}}
     {m},
\label{eq:vehicle_coefficients_1}\\
a_{v3}
&:=
-\frac{2C_{\alpha f}L_f-2C_{\alpha r}L_r}
       {mV_x},
&
b_v
&:=
\frac{2C_{\alpha f}}{m},
\label{eq:vehicle_coefficients_2}\\
a_{r1}
&:=
-\frac{2C_{\alpha f}L_f-2C_{\alpha r}L_r}
       {I_zV_x},
&
a_{r2}
&:=
\frac{2C_{\alpha f}L_f-2C_{\alpha r}L_r}
     {I_z},
\label{eq:vehicle_coefficients_3}\\
a_{r3}
&:=
-\frac{
2C_{\alpha f}L_f^2+
2C_{\alpha r}L_r^2
}{
I_zV_x
},
&
b_r
&:=
\frac{2C_{\alpha f}L_f}{I_z}.
\label{eq:vehicle_coefficients_4}
\end{align}

Using these coefficients, the state matrix takes the compact form
\begin{equation}
\Am_c
=
\begin{bmatrix}
0 & 1 & 0 & 0\\
0 & a_{v1} & a_{v2} & a_{v3}\\
0 & 0 & 0 & 1\\
0 & a_{r1} & a_{r2} & a_{r3}
\end{bmatrix}.
\label{eq:Ac_vehicle}
\end{equation}

The steering-input and path-curvature matrices are
\begin{equation}
\Bm_c
=
\begin{bmatrix}
0\\
b_v\\
0\\
b_r
\end{bmatrix},
\qquad
\Em_c
=
\begin{bmatrix}
0\\
a_{v3}-V_x\\
0\\
a_{r3}
\end{bmatrix}.
\label{eq:BcEc_vehicle}
\end{equation}


\begin{remark}
The term
\(
\Em_c d(t)
\)
is not an arbitrary disturbance but a known reference forcing induced by the road geometry. In straight segments,
\(
\kappa=0
\),
and the model reduces to ordinary lateral-error regulation; in curved segments, steering is required not only to correct existing tracking errors but also to anticipate the yaw motion imposed by the path curvature. Nominal feedforward may compensate this geometric effect, while MPC remains responsible for feedback correction, preview-based optimization, and enforcement of steering, steering-rate, and state constraints. Thus, feedforward compensation and predictive control are complementary rather than redundant.
\end{remark}


Under zero-order hold over each sampling interval
\(
[t_k,t_{k+1})
\),
with
\(
\uv(t)=\uv_k
\)
and
\(
d(t)=d_k
\), Let
\(
T_s:=t_{k+1}-t_k>0
\)
denote the prescribed controller sampling period.
The continuous-time model admits the exact discretization
\(
\xv(t_{k+1})
=
e^{\Am_cT_s}\xv(t_k)
+
\int_{0}^{T_s}e^{\Am_c\tau}\Bm_c\,d\tau\,\uv_k
+
\int_{0}^{T_s}e^{\Am_c\tau}\Em_c\,d\tau\,d_k
\).
Defining
\(
\Am_p:=e^{\Am_cT_s}
\),
\(
\Bm_p:=\int_{0}^{T_s}e^{\Am_c\tau}\Bm_c\,d\tau
\),
and
\(
\Em_p:=\int_{0}^{T_s}e^{\Am_c\tau}\Em_c\,d\tau
\),
with
\(
\xv_k:=\xv(t_k)
\),
the exact discrete-time model is
\begin{equation}
\xv_{k+1}
=
\Am_p\xv_k
+
\Bm_p\uv_k
+
\Em_p d_k.
\label{eq:discrete_vehicle_model}
\end{equation}

In the certification derivation below, \(d_k\) is either compensated
through the nominal feedforward term or treated as a known exogenous
command around which the local regulation model is formed. The
resulting regulation form is
\begin{equation}
\xv_{k+1}
=
\Am_p\xv_k
+
\Bm_p\uv_k.
\label{eq:regulation_vehicle_model}
\end{equation}

\section{Expert MPC and Exact Finite-Horizon Q-Function}
\label{sec:exact_q}

Inspired by the exact Q-loss formulation introduced
in~\cite{Ghezzi2023ExactQ}, this section constructs a
path-referenced MPC expert and its associated exact
finite-horizon Q-function for autonomous-vehicle lateral
control, providing the foundation for the proposed
certified imitation-learning framework.

\subsection{Expert MPC Problem}
\label{subsec:expert_mpc}

Let
\(
N\in\mathbb{N}_{>0}
\)
denote the prediction-horizon length,
\(
\bar{\xv}
\)
the measured current error state,
\(
\uv_{-1}
\)
the previous steering command, and
\(
d_{0:N-1}:=(d_0,\ldots,d_{N-1})
\)
the known curvature preview. Define the MPC context
\(
\bar{\xi}:=(\bar{\xv},\uv_{-1},d_{0:N-1})
\)
and
\(
\Delta\uv_k:=\uv_k-\uv_{k-1}
\).
The stage and terminal costs are


\begin{equation}
\begin{aligned}
\ell_k
&:=
\norm{\xv_k}_{\Qm}^{2}
+
\norm{\uv_k}_{\Rm}^{2}
+
\norm{\Delta\uv_k}_{\Sm}^{2}
\\
&\quad
+
\rho_s\mathbf{1}\tran\sv_k
+
\norm{\sv_k}_{\Zm_s}^{2},
\\[1mm]
V_f(\xv_N,\sv_f)
&:=
\norm{\xv_N}_{\Qm_f}^{2}
+
\rho_f\mathbf{1}\tran\sv_f
\\
&\quad
+
\norm{\sv_f}_{\Zm_f}^{2}.
\end{aligned}
\label{eq:expert_costs}
\end{equation}

where
\(
\Qm,\Qm_f,\Sm,\Zm_s,\Zm_f\succeq0
\),
\(
\Rm\succ0
\),
and
\(
\rho_s,\rho_f>0
\).

The expert MPC solves
\begingroup
\setlength{\jot}{2pt}
\begin{subequations}
\label{eq:expert_ocp}

\begin{equation}
J^\star(\bar{\xi})
=
\min_{\substack{
\xv_{0:N},\,\uv_{0:N-1}\\
\sv_{0:N-1},\,\sv_f}}
\left\{
\sum_{k=0}^{N-1}\ell_k
+
V_f(\xv_N,\sv_f)
\right\},
\label{eq:expert_ocp_cost}
\end{equation}

\begin{equation}
\xv_0=\bar{\xv},
\qquad
\xv_{k+1}
=
\Am_p\xv_k+\Bm_p\uv_k+\Em_p d_k,
\label{eq:expert_ocp_dynamics}
\end{equation}

\begin{equation}
\uv_{\min}
\leq
\uv_k
\leq
\uv_{\max},
\qquad
\Delta\uv_{\min}
\leq
\Delta\uv_k
\leq
\Delta\uv_{\max},
\label{eq:expert_ocp_input}
\end{equation}

\begin{equation}
\Hm_x\xv_k
\leq
\hv_x+\sv_k,
\qquad
\sv_k\succeq\mathbf0,
\label{eq:expert_ocp_state}
\end{equation}

\begin{equation}
H_f\xv_N
\leq
h_f+\sv_f,
\qquad
\sv_f\succeq\mathbf0,
\label{eq:expert_ocp_terminal}
\end{equation}

\end{subequations}
\endgroup
where
\(
k=0,\ldots,N-1
\)
in \eqref{eq:expert_ocp_dynamics}--\eqref{eq:expert_ocp_state}.
The steering and steering-rate bounds are hard, whereas
\(
\sv_k
\)
and
\(
\sv_f
\)
provide penalized relaxation of the state and terminal constraints.
The nominal terminal set is
\begin{equation}
\calX_f
:=
\setc{
\xv\in\R^{n_x}:H_f\xv\leq h_f
}.
\label{eq:terminal_set}
\end{equation}


If
\(
\uv_{0:N-1}^{\star}(\bar{\xi})
=
(\uv_0^\star,\ldots,\uv_{N-1}^\star)
\)
denotes the optimal steering sequence solving
\eqref{eq:expert_ocp}, then the expert MPC policy
\(
\pi_E
\)
is defined by the first receding-horizon input as
\begin{equation}
\pi_E(\bar{\xi})
:=
\uv_0^\star(\bar{\xi}).
\label{eq:expert_policy}
\end{equation}

\subsection{Exact MPC-Induced Q-Function}
\label{subsec:exact_q}


The conventional imitation target
\(
\pi_E(\bar{\xi})
\)
specifies which steering action the expert selects, but it does not
quantify the physical consequence of choosing a different first action.
To recover this information, let
\(
\bar{\uv}
\)
be a feasible candidate steering action and constrain the first MPC
input to
\(
\uv_0=\bar{\uv}
\).
The remaining horizon is then re-optimized by the expert MPC.
Thus, the resulting cost isolates the downstream consequence of taking
\(
\bar{\uv}
\)
at the current step while allowing the controller to recover optimally
thereafter. This construction defines the exact MPC-induced
\(Q\)-function associated with the candidate first action.

The resulting exact MPC-induced \(Q\)-function is
\begingroup
\setlength{\jot}{2pt}
\begin{subequations}
\label{eq:exact_q_ocp}

\begin{equation}
Q_E(\bar{\xi},\bar{\uv})
=
\min_{\substack{
\xv_{0:N},\,\uv_{0:N-1}\\
\sv_{0:N-1},\,\sv_f}}
\left\{
\sum_{k=0}^{N-1}\ell_k
+
V_f(\xv_N,\sv_f)
\right\}.
\label{eq:exact_q_cost}
\end{equation}

\begin{equation}
\xv_0=\bar{\xv},
\qquad
\uv_0=\bar{\uv}.
\label{eq:exact_q_initial}
\end{equation}

\begin{equation}
\begin{aligned}
\xv_{k+1}
&=
\Am_p\xv_k
+
\Bm_p\uv_k
+
\Em_p d_k,
\\[-0.1em]
&\hspace{2em}
k=0,\ldots,N-1.
\end{aligned}
\label{eq:exact_q_dynamics}
\end{equation}

\begin{equation}
\uv_{\min}
\leq
\uv_k
\leq
\uv_{\max},
\qquad
k=0,\ldots,N-1.
\label{eq:exact_q_input}
\end{equation}

\begin{equation}
\Delta\uv_{\min}
\leq
\Delta\uv_k
\leq
\Delta\uv_{\max},
\qquad
k=0,\ldots,N-1.
\label{eq:exact_q_rate}
\end{equation}

\begin{equation}
\begin{aligned}
\Hm_x\xv_k
&\leq
\hv_x+\sv_k,
\qquad
\sv_k\succeq\mathbf{0},
\\[-0.1em]
&\hspace{2em}
k=0,\ldots,N-1.
\end{aligned}
\label{eq:exact_q_state}
\end{equation}

\begin{equation}
H_f\xv_N
\leq
h_f+\sv_f,
\qquad
\sv_f\succeq\mathbf{0}.
\label{eq:exact_q_terminal}
\end{equation}

\end{subequations}
\endgroup

Thus,
\(
Q_E(\bar{\xi},\bar{\uv})
\)
is the minimum finite-horizon cost attainable after forcing the
current steering action to be
\(
\bar{\uv}
\)
and allowing the expert MPC to optimally re-plan the remaining horizon.
It therefore measures the downstream consequence of a candidate
action rather than merely its instantaneous mismatch with the expert.

For the learner
\(
\pi_{\thetav}
\),
define the exact \(Q\)-gap
\begin{equation}
\begin{aligned}
\Delta_Q(\bar{\xi};\thetav)
:={}&
Q_E\bigl(
\bar{\xi},
\pi_{\thetav}(\bar{\xi})
\bigr)
\\[-0.1em]
&-
Q_E\bigl(
\bar{\xi},
\pi_E(\bar{\xi})
\bigr).
\end{aligned}
\label{eq:q_gap}
\end{equation}

\begin{proposition}[Expert action minimizes the exact \(Q\)-function]
\label{prop:q_minimizer}
Suppose \eqref{eq:expert_ocp} is feasible at
\(
\bar{\xi}
\)
and its first optimal input is unique. Then, for every feasible
\(
\bar{\uv}
\),
\begin{equation}
Q_E(\bar{\xi},\bar{\uv})
\geq
Q_E\bigl(
\bar{\xi},
\pi_E(\bar{\xi})
\bigr).
\label{eq:q_min_property}
\end{equation}
Moreover, if
\(
\bar{\uv}\neq\pi_E(\bar{\xi})
\),
then the inequality is strict.
\end{proposition}

\begin{IEEEproof}
The optimization defining
\(
Q_E(\bar{\xi},\bar{\uv})
\)
is identical to the expert MPC problem
\eqref{eq:expert_ocp}, except for the additional equality
\(
\uv_0=\bar{\uv}
\).
Hence, its feasible set is a restriction of the expert feasible set.
For
\(
\bar{\uv}=\pi_E(\bar{\xi})
\),
the expert optimizer remains feasible and attains
\(
J^\star(\bar{\xi})
\).
No other feasible fixed first action can yield a smaller optimum.
Uniqueness of the expert first input gives strictness whenever
\(
\bar{\uv}\neq\pi_E(\bar{\xi})
\).
\end{IEEEproof}

Consequently, whenever the learner-induced fixed-action problem is
feasible,
\begin{equation}
\Delta_Q(\bar{\xi};\thetav)\geq0,
\qquad
\Delta_Q(\bar{\xi};\thetav)=0
\iff
\pi_{\thetav}(\bar{\xi})
=
\pi_E(\bar{\xi}).
\label{eq:q_gap_nonnegative}
\end{equation}

\begin{remark}[Why exact \(Q\) rather than instantaneous imitation]
\label{rem:q_vs_instantaneous}
Let
\(
\delta\uv
:=
\pi_{\thetav}(\bar{\xi})
-
\pi_E(\bar{\xi})
\),
where
\(
\pi_{\thetav}
\)
denotes the learner steering policy parameterized by
\(
\thetav
\).
An instantaneous imitation loss penalizes
\(
\norm{\delta\uv}_{\Wm_u}^{2}
\)
through the fixed metric
\(
\Wm_u
\),
independent of the downstream vehicle response.
In contrast, if
\(
Q_E(\bar{\xi},\uv)
\)
is twice differentiable and
\(
\pi_E(\bar{\xi})
\)
is an interior minimizer, then locally
\(
\Delta_Q(\bar{\xi};\thetav)
=
\frac{1}{2}
\delta\uv\tran
\Hm_Q(\bar{\xi})
\delta\uv
+
o(\norm{\delta\uv}^{2})
\),
where
\(
\Hm_Q(\bar{\xi})
:=
\left.
\nabla_{\uv\uv}^{2}
Q_E(\bar{\xi},\uv)
\right|_{\uv=\pi_E(\bar{\xi})}
\succeq0
\).
Hence,
\(
\frac{1}{2}
\lambda_{\min}(\Hm_Q)
\norm{\delta\uv}^{2}
\lesssim
\Delta_Q
\lesssim
\frac{1}{2}
\lambda_{\max}(\Hm_Q)
\norm{\delta\uv}^{2}
\),
so equal steering mismatches need not have equal physical consequences. In the unconstrained quadratic regime, let the optimal cost-to-go from
the next predicted state be locally written as
\(
V_1(\xv_1)
=
\xv_1\tran\Pm_1\xv_1
+
2\pv_1\tran\xv_1
+
c_1
\),
where
\(
\Pm_1\succeq0
\)
is the quadratic cost-to-go matrix and
\(
\xv_1
=
\Am_p\bar{\xv}
+
\Bm_p\uv
+
\Em_p d_0
\).
Up to terms independent of
\(
\uv
\),
the fixed-action value therefore contains
\(
\uv\tran\Rm\uv
+
(\uv-\uv_{-1})\tran\Sm(\uv-\uv_{-1})
+
V_1(\xv_1)
\).
Its quadratic dependence on
\(
\uv
\)
is consequently
\(
\uv\tran
\bigl(
\Rm+\Sm+\Bm_p\tran\Pm_1\Bm_p
\bigr)
\uv
\).
Since the Hessian of
\(
\uv\tran\Mm\uv
\)
with respect to
\(
\uv
\)
is
\(
2\Mm
\),
it follows that
\(
\frac{1}{2}\Hm_Q
=
\Rm+\Sm+\Bm_p\tran\Pm_1\Bm_p
\).
Here,
\(
\Bm_p\tran\Pm_1\Bm_p
\)
captures the future lateral--yaw cost induced by the present steering
action. The resulting metric may become state- and context-dependent
when the local MPC model, weighting, admissible set, or active
constraint set changes.
\end{remark}

\section{Exact-Q Hybrid Imitation Objective}
\label{sec:hybrid_objective}

Let \(\calD\) denote a dataset collected from expert rollouts,
learner-induced rollouts, or both, with the corresponding MPC context
\(
(\uv_{-1},d_0,\ldots,d_{N-1})
\)
attached to each queried state but suppressed from the notation.

The conventional imitation loss is
\begin{equation}
\calL_{\mathrm{im}}^{\calD}(\thetav)
=
\bbE_{\xv\sim\calD}
\left[
\norm{
\pi_{\thetav}(\xv)-\pi_E(\xv)
}_{\Wm_u}^{2}
\right],
\qquad
\Wm_u=\Wm_u\tran\succeq0.
\label{eq:imitation_loss}
\end{equation}

To penalize actions according to their finite-horizon consequence, we
use
\begin{equation}
\calL_Q^{\calD}(\thetav)
=
\bbE_{\xv\sim\calD}
\left[
\Delta_Q(\xv;\thetav)
\right].
\label{eq:q_loss}
\end{equation}

The proposed hybrid objective is
\begin{equation}
\calL_{\mathrm{hyb}}^{\calD}(\thetav)
=
\alpha
\calL_{\mathrm{im}}^{\calD}(\thetav)
+
\beta
\calL_Q^{\calD}(\thetav),
\qquad
\alpha,\beta\geq0.
\label{eq:hybrid_loss}
\end{equation}

The two terms play complementary roles:
\(\calL_{\mathrm{im}}^{\calD}\) preserves local agreement with the
expert steering map, whereas
\(\calL_Q^{\calD}\) weights a learner action by its actual
finite-horizon control consequence.
This distinction is particularly relevant in lateral control because
equal steering mismatches need not be equally harmful; near tight
curves or active constraints, a small action error may induce a much
larger downstream lateral or yaw deviation.

\begin{remark}[Clarification of the context-dependent exact-\(Q\) metric]
\label{rem:q_metric_context}
The context dependence discussed in
Remark~\ref{rem:q_vs_instantaneous}
should be understood locally.
Under the unconstrained LTI quadratic representation used above,
\(
\frac{1}{2}\Hm_Q
=
\Rm+\Sm+\Bm_p\tran\Pm_1\Bm_p
\).
Hence, if
\(
\Rm
\),
\(
\Sm
\),
\(
\Bm_p
\),
and
\(
\Pm_1
\)
remain fixed, then
\(
\Hm_Q
\)
is fixed as well.
Changes in
\(
\bar{\xi}
=
(\bar{\xv},\uv_{-1},d_{0:N-1})
\)
may then shift the expert action
\(
\pi_E(\bar{\xi})
\)
and the corresponding optimal cost without changing the local
curvature of
\(
Q_E
\).
In particular, when the preview enters only through the additive term
\(
\Em_p d_k
\),
curvature preview alone does not change
\(
\Hm_Q
\)
within a fixed unconstrained quadratic regime.

For constrained MPC, the parametric value function is locally
piecewise quadratic. Let
\(
i
\)
and
\(
j
\)
denote two local critical regions associated with different active
constraint sets
\(
\mathcal A_i
\)
and
\(
\mathcal A_j
\),
respectively.
Within each region, the corresponding exact-\(Q\) function has a local
quadratic curvature
\(
\Hm_Q^{(i)}
\)
or
\(
\Hm_Q^{(j)}
\).
If the active set changes, the local quadratic value function may also
change, so that
\(
\Pm_1^{(i)}
\neq
\Pm_1^{(j)}
\)
and, consequently,
\(
\Hm_Q^{(i)}
=
2\bigl(
\Rm+\Sm+\Bm_p\tran\Pm_1^{(i)}\Bm_p
\bigr)
\)
may differ from
\(
\Hm_Q^{(j)}
=
2\bigl(
\Rm+\Sm+\Bm_p\tran\Pm_1^{(j)}\Bm_p
\bigr)
\).
Likewise, variations of the local plant model, weighting matrices, or
admissible set may modify
\(
\Bm_p
\),
\(
\Pm_1
\),
or both, and therefore modify
\(
\Hm_Q
\).

Thus, exact \(Q\) does not assign different penalties merely because
the state or curvature preview changes; rather, it inherits whatever
local action sensitivity is induced by the underlying MPC problem.
Equal steering mismatches receive different finite-horizon penalties
when the local optimization geometry, and hence the curvature of
\(
Q_E
\),
actually changes.
\end{remark}

\subsection{Differentiation Through the Exact Q-Function}
\label{subsec:q_gradient}

Let the fixed-action constraint in \eqref{eq:exact_q_ocp} be
\begin{equation}
\cv_0(\uv_0,\bar{\uv})
:=
\uv_0-\bar{\uv}
=
\mathbf{0},
\label{eq:fixed_action_constraint}
\end{equation}
with Lagrange multiplier \(\lambdav_0\).
Under local uniqueness and the standard strong-regularity conditions
of Linear Independence Constraint Qualification (LICQ), strict complementarity, and second-order sufficiency, the
envelope theorem yields
\begin{equation}
\nabla_{\bar{\uv}}
Q_E(\bar{\xv},\bar{\uv})
=
-\lambdav_0.
\label{eq:q_gradient_multiplier}
\end{equation}

Hence, by the chain rule,
\begin{equation}
\begin{aligned}
\nabla_{\thetav}
\calL_Q^{\calD}(\thetav)
={}&
\bbE_{\xv\sim\calD}
\Bigg[
\nabla_{\thetav}
\pi_{\thetav}(\xv)\tran
\\[-0.1em]
&\qquad\cdot
\left.
\nabla_{\uv}
Q_E(\xv,\uv)
\right|_{\uv=\pi_{\thetav}(\xv)}
\Bigg].
\end{aligned}
\label{eq:q_gradient_theta}
\end{equation}

Similarly,
\begin{equation}
\begin{aligned}
\nabla_{\thetav}
\calL_{\mathrm{hyb}}^{\calD}(\thetav)
={}&
\bbE_{\xv\sim\calD}
\Bigg[
2\alpha
\nabla_{\thetav}
\pi_{\thetav}(\xv)\tran
\Wm_u
\\
&\qquad\cdot
\bigl(
\pi_{\thetav}(\xv)-\pi_E(\xv)
\bigr)
\\
&\qquad+
\beta
\nabla_{\thetav}
\pi_{\thetav}(\xv)\tran
\\
&\qquad\cdot
\left.
\nabla_{\uv}
Q_E(\xv,\uv)
\right|_{\uv=\pi_{\thetav}(\xv)}
\Bigg].
\end{aligned}
\label{eq:hybrid_gradient_theta}
\end{equation}

Equation \eqref{eq:q_gradient_multiplier} provides the direct link
between the MPC expert and the learner: the multiplier associated with
the fixed first action quantifies the marginal increase in optimal
finite-horizon cost caused by perturbing the learner steering command.

\section{LFT Representation of the Neural Controller}


Consider an \(L\)-hidden-layer feedforward policy. Let
\(\zv_0=\xv\). For \(\ell=0,\ldots,L-1\),
\begin{equation}
\vv_\ell=\Wm_\ell\zv_\ell+\bv_\ell,\qquad
\zv_{\ell+1}=\bm{\phi}_\ell(\vv_\ell),
\end{equation}
and the output layer is
\begin{equation}
\uv=\Wm_L\zv_L+\bv_L.
\end{equation}


For equilibrium certification, let
\(
\vv_\ell^\star
\),
\(
\zv_\ell^\star
\),
and
\(
\uv^\star=\pi_{\thetav}(\mathbf{0})
\)
denote the internal and output equilibrium values induced by
\(
\xv=\mathbf{0}
\).
Define the deviation variables
\(
\widetilde{\vv}_\ell:=\vv_\ell-\vv_\ell^\star
\),
\(
\widetilde{\zv}_\ell:=\zv_\ell-\zv_\ell^\star
\),
and
\(
\widetilde{\uv}:=\uv-\uv^\star
\).
Subtracting the equilibrium network equations eliminates the biases and
gives
\(
\widetilde{\vv}_0=\Wm_0\xv
\),
\(
\widetilde{\vv}_\ell=\Wm_\ell\widetilde{\zv}_\ell
\)
for
\(
\ell=1,\ldots,L-1
\),
and
\(
\widetilde{\uv}=\Wm_L\widetilde{\zv}_L
\).
Moreover,
\(
\widetilde{\zv}_{\ell+1}
=
\widetilde{\phi}_\ell(\widetilde{\vv}_\ell)
\),
where
\(
\widetilde{\phi}_\ell(\eta)
:=
\phi_\ell(\vv_\ell^\star+\eta)
-
\phi_\ell(\vv_\ell^\star)
\)
satisfies
\(
\widetilde{\phi}_\ell(\mathbf{0})=\mathbf{0}
\).
Stack
\(
\widetilde{\vv}
=
\colop(\widetilde{\vv}_0,\ldots,\widetilde{\vv}_{L-1})
\)
and
\(
\widetilde{\wv}
=
\colop(\widetilde{\zv}_1,\ldots,\widetilde{\zv}_L)
\).
For notational simplicity, the tildes are omitted hereafter, so that the
equilibrium-shifted neural policy can be written as the following
LFT-like interconnection.
\begin{subequations}
\label{eq:lft_controller}
\begin{align}
\vv_k&=\Am_\pi(\thetav)\wv_k+\Bm_\pi(\thetav)\xv_k,\\
\wv_k&=\bm{\phi}(\vv_k),\\
\uv_k&=\Cm_\pi(\thetav)\wv_k+\Dm_\pi(\thetav)\xv_k.
\end{align}
\end{subequations}

For the general \(L\)-hidden-layer feedforward architecture,
\begin{equation}
\Am_\pi(\thetav)
=
\begin{bmatrix}
\mathbf{0} & \mathbf{0} & \mathbf{0} & \cdots & \mathbf{0}\\
\Wm_1      & \mathbf{0} & \mathbf{0} & \cdots & \mathbf{0}\\
\mathbf{0} & \Wm_2      & \mathbf{0} & \cdots & \mathbf{0}\\
\vdots     & \ddots     & \ddots     & \ddots & \vdots\\
\mathbf{0} & \cdots     & \mathbf{0} & \Wm_{L-1} & \mathbf{0}
\end{bmatrix}.
\label{eq:lft_general_A}
\end{equation}

The remaining matrices are compactly written as
\begin{equation}
\begin{aligned}
\Bm_\pi(\thetav)
&=
\begin{bmatrix}
\Wm_0\tran &
\mathbf{0} &
\cdots &
\mathbf{0}
\end{bmatrix}\tran,
\\
\Cm_\pi(\thetav)
&=
\begin{bmatrix}
\mathbf{0} &
\cdots &
\mathbf{0} &
\Wm_L
\end{bmatrix},
\qquad
\Dm_\pi(\thetav)=\mathbf{0}.
\end{aligned}
\label{eq:lft_general_BCD}
\end{equation}

Equivalently, the block entries of \(\Am_\pi(\thetav)\) satisfy
\begin{equation}
\big[\Am_\pi(\thetav)\big]_{ij}
=
\begin{cases}
\Wm_i, & i=j+1,\\
\mathbf{0}, & \text{otherwise},
\end{cases}
\qquad
i,j=1,\ldots,L.
\label{eq:lft_general_block_entries}
\end{equation}
Hence, \(\Am_\pi(\thetav)\) is strictly block-lower triangular and,
more specifically,
\begin{equation}
\Am_\pi(\thetav)^L=\mathbf{0}.
\end{equation}
Thus, \(\Am_\pi(\thetav)\) is nilpotent, and the feedforward neural
interconnection in \eqref{eq:lft_controller} is well posed.

Substituting \eqref{eq:lft_controller} into the regulation form \eqref{eq:regulation_vehicle_model} gives the closed-loop Lur'e system
\begin{subequations}
\label{eq:lure_system}
\begin{align}
\xv_{k+1}&=\Am_c^\theta\xv_k+\Bm_c^\theta\wv_k,\\
\vv_k&=\Bm_\pi(\thetav)\xv_k+\Am_\pi(\thetav)\wv_k,\\
\wv_k&=\bm{\phi}(\vv_k),
\end{align}
\end{subequations}
where
\begin{equation}
\Am_c^\theta=\Am_p+\Bm_p\Dm_\pi(\thetav),\qquad
\Bm_c^\theta=\Bm_p\Cm_\pi(\thetav).\label{eq:closed_loop_matrices}
\end{equation}

\section{Lyapunov--IQC Certificate}

Assume each activation is slope-restricted in $[0,1]$:
\begin{equation}
0\leq\frac{\phi_i(a)-\phi_i(b)}{a-b}\leq1,\qquad a\neq b.
\end{equation}
If $\phi_i(0)=0$, then $w_i(v_i-w_i)\geq0$. For a diagonal multiplier
\begin{equation}
\Lam=\diagop(\lambda_1,\ldots,\lambda_m)\succ0,
\end{equation}
the stacked sector IQC is
\begin{equation}
2\wv_k\tran\Lam(\vv_k-\wv_k)\geq0.
\end{equation}
Equivalently,
\begin{equation}
\begin{bmatrix}\vv_k\\ \wv_k\end{bmatrix}\tran
\Mm_\phi(\Lam)
\begin{bmatrix}\vv_k\\ \wv_k\end{bmatrix}
\geq0,\qquad
\Mm_\phi(\Lam)=
\begin{bmatrix}
\mathbf{0} & \Lam\\
\Lam & -2\Lam
\end{bmatrix}.
\end{equation}
Define
\begin{equation}
\xiv_k=\begin{bmatrix}\xv_k\\ \wv_k\end{bmatrix},\qquad
\Sm_\pi(\thetav)=
\begin{bmatrix}
\Bm_\pi(\thetav) & \Am_\pi(\thetav)\\
\mathbf{0} & \Id
\end{bmatrix}.
\end{equation}
Since $[\vv_k\tran\ \wv_k\tran]\tran=\Sm_\pi(\thetav)\xiv_k$, all admissible activation trajectories satisfy
\begin{equation}
\xiv_k\tran\Sm_\pi(\thetav)\tran\Mm_\phi(\Lam)\Sm_\pi(\thetav)\xiv_k\geq0.\label{eq:xi_iqc}
\end{equation}

Let $V(\xv)=\xv\tran\Pm\xv$, where $\Pm=\Pm\tran\succ0$. From \eqref{eq:lure_system},
\begin{equation}
V(\xv_{k+1})-V(\xv_k)=\xiv_k\tran\Hm(\Pm,\thetav)\xiv_k,
\end{equation}
with
\begin{align}
\Hm(\Pm,\thetav)
&=
\begin{bmatrix}
\Am_c^\theta & \Bm_c^\theta
\end{bmatrix}\tran
\Pm
\begin{bmatrix}
\Am_c^\theta & \Bm_c^\theta
\end{bmatrix}
-
\begin{bmatrix}
\Id & \mathbf{0}
\end{bmatrix}\tran
\Pm
\begin{bmatrix}
\Id & \mathbf{0}
\end{bmatrix}.
\end{align}
The Lyapunov--IQC stability matrix is
\begin{equation}
\calM(\Pm,\Lam,\thetav)
=
\Hm(\Pm,\thetav)
+
\Sm_\pi(\thetav)\tran\Mm_\phi(\Lam)\Sm_\pi(\thetav).\label{eq:stability_matrix}
\end{equation}

\begin{theorem}[Lyapunov--IQC stability]
Consider \eqref{eq:lure_system}. Suppose the activation satisfies \eqref{eq:xi_iqc}. If there exist $\Pm=\Pm\tran\succ0$ and $\Lam\succ0$ such that
\begin{equation}
\calM(\Pm,\Lam,\thetav)\prec0,
\end{equation}
then the origin of the neural closed-loop system is asymptotically stable.
\end{theorem}

\begin{IEEEproof}
For every admissible activation trajectory, \eqref{eq:xi_iqc} holds. If $\calM(\Pm,\Lam,\thetav)\prec0$, then for every nonzero $\xiv_k$,
\begin{equation}
\xiv_k\tran\left[
\Hm(\Pm,\thetav)+
\Sm_\pi(\thetav)\tran\Mm_\phi(\Lam)\Sm_\pi(\thetav)
\right]\xiv_k<0.
\end{equation}
The IQC term is nonnegative along the trajectory, hence $\xiv_k\tran\Hm(\Pm,\thetav)\xiv_k<0$. Therefore $V(\xv_{k+1})-V(\xv_k)<0$ for $\xv_k\neq\mathbf{0}$, and $\Pm\succ0$ proves asymptotic stability.
\end{IEEEproof}

Let
\begin{equation}
\Am=\Am_c^\theta,\quad \Bm=\Bm_c^\theta,\quad
\Em=\Bm_\pi(\thetav),\quad \Fm=\Am_\pi(\thetav).
\end{equation}
Then the matrix in \eqref{eq:stability_matrix} expands to
\begin{equation}
\calM=
\begin{bmatrix}
\Am\tran\Pm\Am-\Pm
&
\Am\tran\Pm\Bm+\Em\tran\Lam\\
\Bm\tran\Pm\Am+\Lam\Em
&
\Bm\tran\Pm\Bm+\Fm\tran\Lam+\Lam\Fm-2\Lam
\end{bmatrix}.\label{eq:M_block}
\end{equation}

The certified set is
\begin{equation}
\calS_{\mathrm{cert}}
=
\left\{
\thetav:
\exists \Pm=\Pm\tran\succ0,\ \Lam\succ0
\ \mathrm{s.t.}\ 
\calM(\Pm,\Lam,\thetav)\prec0
\right\}.
\end{equation}
The certificate margin is defined by
\begin{equation}
\mu(\Pm,\Lam,\thetav)=
-\lmax\left(\calM(\Pm,\Lam,\thetav)\right).
\end{equation}
A positive $\mu$ means the controller is certified. The logarithmic barrier used during training is
\begin{equation}
\calB_{\mathrm{stab}}(\Pm,\Lam,\thetav)
=
-\log\left(
-\lmax\left(\calM(\Pm,\Lam,\thetav)\right)
\right).\label{eq:barrier}
\end{equation}

\section{Certified Exact-Q Training}
\label{sec7}

The complete certified objective is
\begin{equation}
\calL_{\mathrm{cert}}^{\calD}(\thetav,\Pm,\Lam)
=
\alpha\calL_{\mathrm{im}}^{\calD}(\thetav)
+
\beta\calL_Q^{\calD}(\thetav)
+
\rho\calB_{\mathrm{stab}}(\Pm,\Lam,\thetav).\label{eq:cert_loss}
\end{equation}
The optimization problem is
\begin{align}
\min_{\thetav,\Pm,\Lam}\quad&
\calL_{\mathrm{cert}}^{\calD}(\thetav,\Pm,\Lam)\\
\mathrm{s.t.}\quad&
\Pm=\Pm\tran\succ0,\qquad \Lam\succ0.
\end{align}
To obtain an unconstrained parametrization, set
\begin{equation}
\Pm=\Lm\Lm\tran+\varepsilon_P\Id,\qquad
\Lam=\diagop(e^{\lambda_1},\ldots,e^{\lambda_m})
+\varepsilon_\Lambda\Id.
\end{equation}
where $\varepsilon_P,\varepsilon_\Lambda>0$. Training then minimizes
\begin{equation}
\min_{\thetav,\Lm,\lambdav}
\calL_{\mathrm{cert}}^{\calD}
\left(
\thetav,\Lm\Lm\tran+\varepsilon_P\Id,
\diagop(e^{\lambda_1},\ldots,e^{\lambda_m})
+\varepsilon_\Lambda\Id
\right).
\end{equation}

The barrier term grows rapidly as $\lmax(\calM)$ approaches zero from below. Consequently, the optimizer is discouraged from moving toward uncertified neural policies. This is important in the vehicle experiment because a visually accurate imitation network can still generate steering actions that destabilize the yaw dynamics when deployed outside the expert distribution.

A nominal exact-$Q$ update without certification is
\begin{equation}
\thetav_{t+1}^{\mathrm{nom}}
=
\thetav_t-\eta
\nabla_{\thetav}
\left[
\alpha\calL_{\mathrm{im}}^{\calD_t}(\thetav_t)+
\beta\calL_Q^{\calD_t}(\thetav_t)
\right].
\end{equation}

If this update violates the certificate, a safe projection is applied:
\begin{equation}
\begin{aligned}
\Pi_{\calS_{\mathrm{cert}}}(\bar{\thetav})
=
\argminop_{\thetav,\Pm,\Lam}\quad&
\norm{\thetav-\bar{\thetav}}^2
\\
\mathrm{s.t.}\quad&
\Pm=\Pm\tran\succ0,
\\
&
\Lam\succ0,
\\
&
\calM(\Pm,\Lam,\thetav)\prec0.
\end{aligned}
\end{equation}

\begin{proposition}[Certification after projection]
If the projection problem is feasible and returns $(\thetav^+,\Pm^+,\Lam^+)$, then $\thetav^+\in\calS_{\mathrm{cert}}$ and the closed loop is asymptotically stable.
\end{proposition}


\section{Certified DAgger}
\label{sec8}

The data distribution generated by the expert is not identical to the
data distribution induced by the learned controller. In particular,
the supervised dataset is typically collected under the expert policy
and therefore reflects the expert state-occupancy distribution, whereas
deployment of
\(
\pi_{\thetav}
\)
induces its own on-policy state distribution. Even when the
instantaneous imitation error is small, each steering mismatch perturbs
the subsequent vehicle state, so the resulting closed-loop trajectory
gradually deviates from the expert trajectory. These deviations
accumulate over time, leading to the well-known compounding-error or
covariate-shift phenomenon in sequential imitation learning. As a
consequence, the learned policy may visit states that are rare or absent
in the original supervised dataset, where its approximation accuracy is
poorly supported and further control errors may be amplified. Thus, a
small pointwise imitation error can induce a progressive distributional
shift between the expert and learner rollouts. To reduce this mismatch,
the training distribution must be augmented with states actually visited
by the learned controller and relabeled by the expert. However, since
such on-policy rollouts directly affect the closed-loop vehicle
trajectory, unrestricted data aggregation may expose the system to
uncertified neural policies. Therefore, to reduce the distribution
mismatch while preserving closed-loop safety, the data aggregation
procedure is restricted to certified policies.

\begin{figure}[h]
    \centering

    \begin{minipage}[c]{0.47\textwidth}
        \centering
        \includegraphics[
            width=\linewidth,
            keepaspectratio
        ]{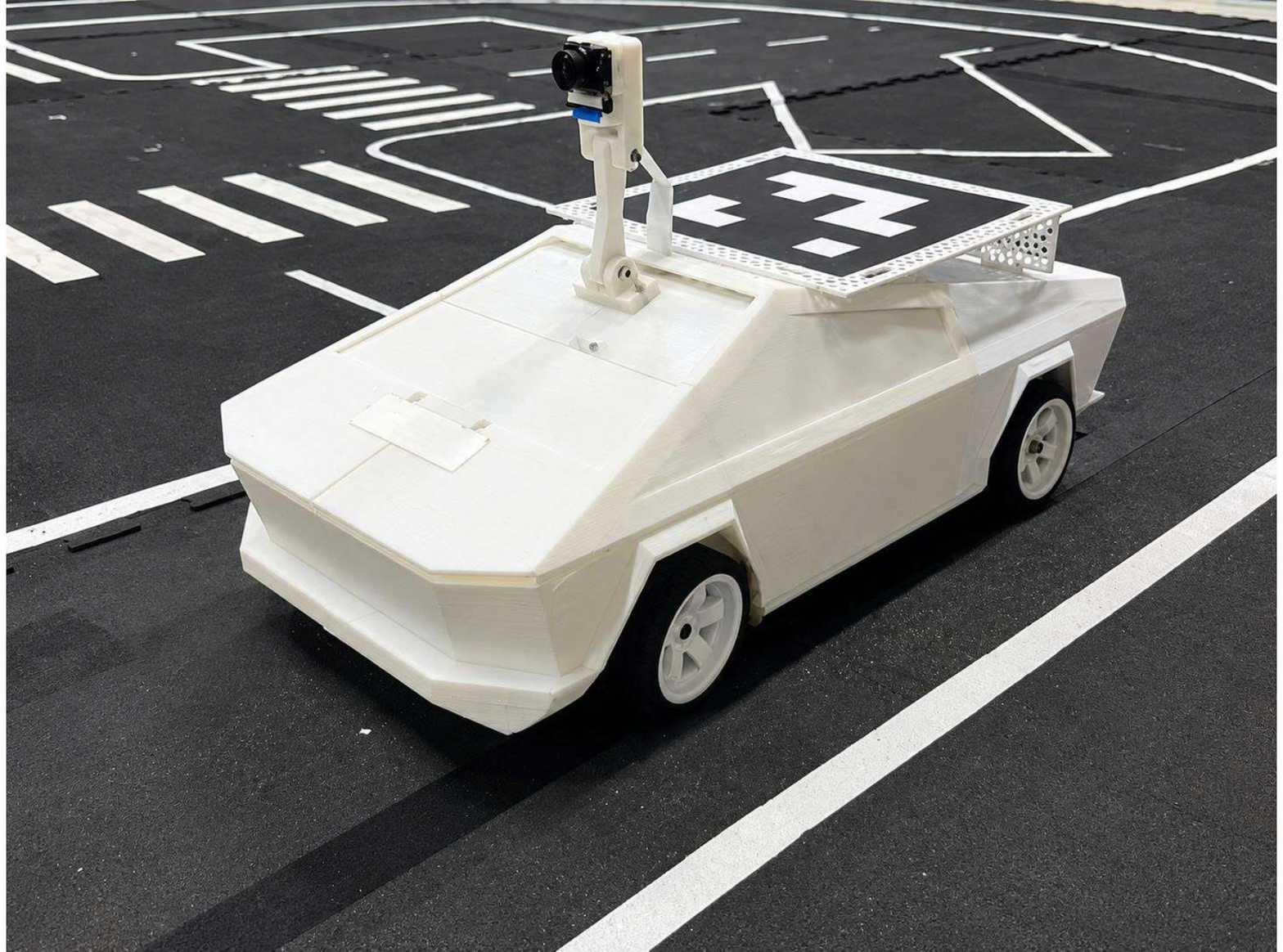}
    \end{minipage}
    \hfill
    \begin{minipage}[c]{0.47\textwidth}
        \centering
        \includegraphics[
            width=\linewidth,
            keepaspectratio
        ]{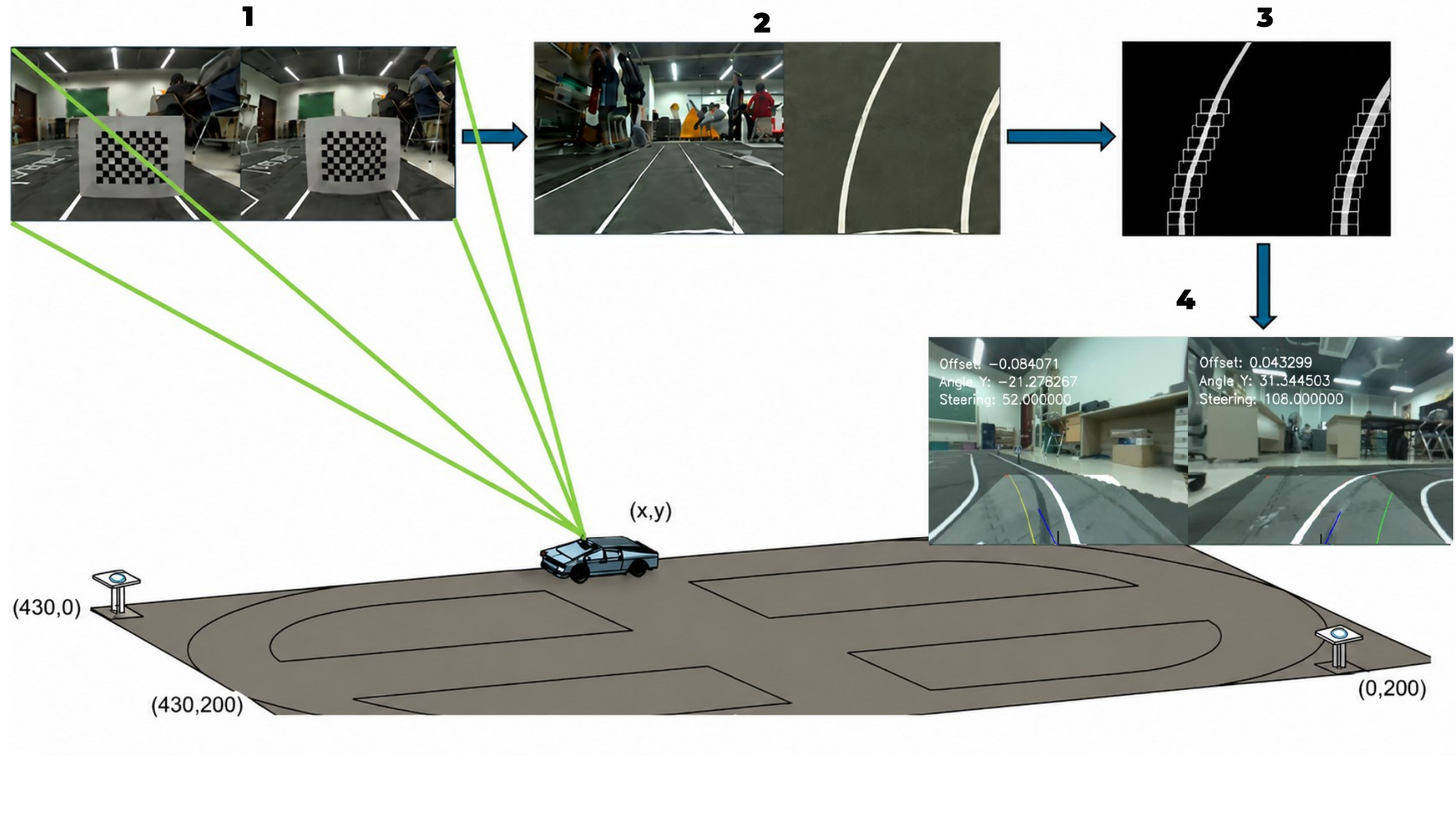}
    \end{minipage}

    \caption{Experimental platform and onboard lane-perception pipeline:
    (top) small-scale autonomous vehicle operating on the indoor track and
    (bottom) vision-processing pipeline used to estimate the lane geometry and
    derive the path-referenced quantities required for lateral control.}
    \label{fig:platform_pipeline}
\end{figure}

At DAgger iteration $i$, the certified policy $\pi_{\thetav_i}$ is deployed:
\begin{equation}
\xv_{k+1}=\Am_p\xv_k+\Bm_p\pi_{\thetav_i}(\xv_k)+\Em_p d_k.
\end{equation}
Let $\calX_i^{\mathrm{roll}}=\{\xv_k^{(i)}\}_{k=0}^{T_i}$ be the states collected during this rollout. The expert is queried on those states and the dataset is updated by
\begin{equation}
\calD_{i+1}
=
\calD_i\cup
\left\{
\left(\xv_k^{(i)},\pi_E(\xv_k^{(i)})\right):
\xv_k^{(i)}\in\calX_i^{\mathrm{roll}}
\right\}.
\end{equation}
The next nominal policy is trained on $\calD_{i+1}$ using the exact-$Q$ hybrid objective, after which barrier training or projection returns the deployed policy. The procedure is summarized in Algorithm~\ref{alg:certified_dagger}.


\begin{lemma}
\label{lem:certified_rollouts}
    Assume $\thetav_0\in\calS_{\mathrm{cert}}$ and each barrier update or projection step returns a policy in $\calS_{\mathrm{cert}}$. Then every rollout used for data aggregation is generated by an asymptotically stable neural closed-loop controller.
\end{lemma}


\begin{figure*}[!t]
    \centering
    \includegraphics[
        width=0.92\textwidth,
        keepaspectratio
    ]{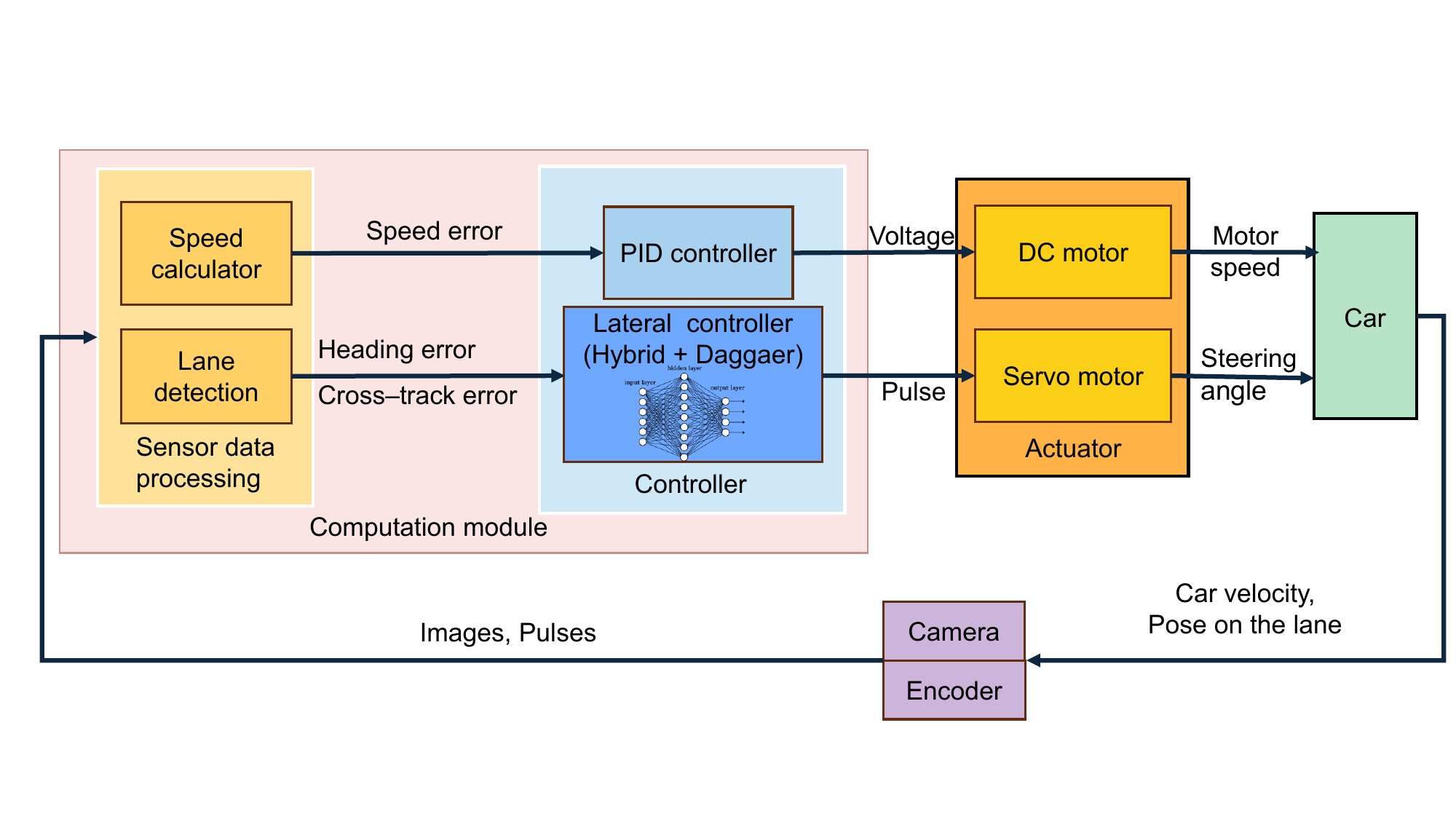}
    \caption{
        Experimental closed-loop architecture. The longitudinal loop uses
        encoder feedback and PID control to regulate the DC drive motor,
        whereas the lateral loop uses either the MPC expert or a learned neural
        policy to command the steering servo.
    }
    \label{fig:control_architecture}
\end{figure*}
\section{Experimental Deployment}

\subsection{Physical platform and control architecture}

The proposed controllers were evaluated on a small-scale autonomous
vehicle operating on a CAD-referenced indoor track. The closed-loop
control architecture is shown in Fig.~\ref{fig:control_architecture}.

For clarity, the controllers compared in the experiments are defined
as follows. MPC denotes the finite-horizon expert developed in
Sec.~\ref{sec:exact_q}. BC denotes the conventional behavior-cloning policy trained
with the instantaneous imitation objective introduced in
Sec.~\ref{sec:hybrid_objective}, whereas Exact-\(Q\) denotes the
task-aware policy trained using the exact MPC-induced \(Q\)-objective
developed in Sec.~\ref{sec:exact_q}. Hybrid denotes the proposed
certified policy trained with the combined objective of Sec.~\ref{sec7},
which incorporates imitation, Exact-\(Q\), and the Lyapunov--IQC
stability barrier. The prefix D+ denotes the corresponding policy after
the certified DAgger refinement of Sec.~\ref{sec8}. Accordingly, BC+D,
D+Exact-\(Q\), and D+Hybrid denote the DAgger-refined versions of BC,
Exact-\(Q\), and Hybrid, respectively.

The architecture consists of separate longitudinal and lateral control
loops. The longitudinal loop regulates vehicle speed using
wheel-encoder feedback and a PID controller that commands the DC
drive motor. The lateral loop generates the steering command applied
to the servo motor and is implemented by either the MPC expert or a
learned neural policy. During deployment, the neural policy replaces
only the MPC lateral controller, while the longitudinal PID loop
remains unchanged.





A forward-facing monocular camera provides the visual input for lane
perception. The captured images are transformed into a bird's-eye
view, from which the lane boundaries and centerline geometry are
estimated. These geometric features are then used to derive the
path-referenced quantities required for lateral control, including
lateral deviation, heading error, and local curvature. The
experimental platform and onboard lane-perception pipeline are shown
in Fig.~\ref{fig:platform_pipeline}.




\subsection{Deployment protocol}


\begin{algorithm}[!htb]
\caption{Certified Exact-$Q$ DAgger}
\label{alg:certified_dagger}
\begin{algorithmic}[1]

\State Construct the path-referenced MPC expert $\pi_E$ 
from \eqref{eq:expert_ocp}.

\State Collect 
$\calD_0=\{(\xi_j,\pi_E(\xi_j))\}_{j=1}^{N_0}$ 
from expert rollouts.

\State Initialize a certified neural policy 
$\pi_{\thetav_0}$ with certificate variables 
$(\Pm_0,\Lam_0)$ such that 
$\thetav_0\in\calS_{\mathrm{cert}}$.

\For{$i=0,\ldots,M-1$}

    \State Deploy only the certified policy 
    $\pi_{\thetav_i}$ on the vehicle or simulator.

    \State Record learner-induced MPC contexts 
    $\calX_i^{\mathrm{roll}}$.

    \State Query the MPC expert $\pi_E(\xi)$ 
    for all $\xi\in\calX_i^{\mathrm{roll}}$.

    \State Aggregate the dataset:
    $\calD_{i+1}\gets\calD_i\cup
    \{(\xi,\pi_E(\xi)):
    \xi\in\calX_i^{\mathrm{roll}}\}$.

    \State Train a nominal policy $\thetav_{\mathrm{nom}}$ 
    using $\calL_{\mathrm{im}}+\beta\calL_Q$.

    \State Search for certificate variables 
    $(\Pm_{\mathrm{nom}},\Lam_{\mathrm{nom}})$.

    \If{the nominal policy satisfies the Lyapunov--IQC certificate}

        \State $(\thetav_{i+1},\Pm_{i+1},\Lam_{i+1})
        \gets
        (\thetav_{\mathrm{nom}},
        \Pm_{\mathrm{nom}},\Lam_{\mathrm{nom}})$.

    \Else

        \State Apply safe projection or barrier training 
        with a finite iteration budget.

        \If{a certified candidate 
        $(\thetav_c,\Pm_c,\Lam_c)$ is obtained}

            \State $(\thetav_{i+1},\Pm_{i+1},\Lam_{i+1})
            \gets(\thetav_c,\Pm_c,\Lam_c)$.

        \Else

            \State $(\thetav_{i+1},\Pm_{i+1},\Lam_{i+1})
            \gets(\thetav_i,\Pm_i,\Lam_i)$.

        \EndIf

    \EndIf

\EndFor

\State \Return the certified neural MPC approximation 
$\pi_{\thetav_M}$ with certificate 
$(\Pm_M,\Lam_M)$.

\end{algorithmic}
\end{algorithm}

Using the experimental platform described above, the MPC controller
serves as the expert for demonstration collection and policy
comparison. Initial MPC rollouts are used to train the BC policy,
which subsequently initializes the Exact-$Q$ and hybrid policies.
These policies are then refined using the objectives introduced in
Sections~V and VIII.

During deployment, the neural policy receives the observation vector
\begin{equation}
\begin{aligned}
\bm{o}_k = \operatorname{col}\big(
& e_y(k),\, e_\psi(k),\,
\kappa_{0,k},\, \kappa_{1,k},\\
& \kappa_{2,k},\, \kappa_{3,k},\,
V_x(k),\, \delta_{k-1}
\big).
\end{aligned}
\label{eq:deployment_observation}
\end{equation}

The corresponding policy and MPC configurations are summarized in
Table~\ref{tab:controller_parameters}.

\begin{table}[!htb]
\centering
\caption{Main controller and deployment parameters.}
\label{tab:controller_parameters}
\scriptsize
\setlength{\tabcolsep}{3.5pt}
\renewcommand{\arraystretch}{0.95}
\begin{tabular}{lc}
\toprule
Parameter & Value \\
\midrule
Policy architecture & $8$--$32$--$32$--$1$ \\
Hidden activation & $\tanh$ \\
Policy input dimension & $8$ \\
Policy curvature preview & $4$ samples \\
MPC prediction horizon $N$ & $10$ \\
MPC curvature preview & $10$ samples \\
MPC weights $(q_1,q_2,r)$ & $(2300,140,12)$ \\
Steering limit & $\pm28^\circ$ \\
Nominal speed $V_x$ & $0.15~\mathrm{m/s}$ \\
Control-command period & $20~\mathrm{ms}$ \\
\bottomrule
\end{tabular}
\end{table}

During DAgger data collection, the MPC expert is evaluated at
learner-visited states to provide expert labels. In the experimental
implementation, it also acts as a safety supervisor when predefined
intervention conditions are violated.

\subsection{Physical closed-loop evaluation}

Physical closed-loop performance is evaluated using the quantitative
metrics reported in Table~\ref{tab:real_results}, which summarize the
main tracking and steering characteristics observed during the
real-vehicle experiments. Lateral and heading errors are obtained from
the externally measured trajectory, whereas steering and Exact-$Q$
quantities are computed from the corresponding onboard control logs.

\begin{table}[H]
\centering
\caption{Real-vehicle closed-loop performance. Lower is better.}
\label{tab:real_results}
\scriptsize
\setlength{\tabcolsep}{2.2pt}
\renewcommand{\arraystretch}{0.95}
\begin{tabular}{lccccc}
\toprule
Controller
& $\mathrm{RMSE}_{y}$
& $\mathrm{RMSE}_{\psi}$
& $\mathrm{MAE}_{\delta}$
& $\overline{\Delta Q}$
& $\mathrm{RMSE}_{\Delta\delta}$ \\
& [m] & [rad] & [deg] & & [deg/step] \\
\midrule
MPC
& 0.0415 & 0.2618 & -- & -- & 6.41 \\
BC
& 0.0565 & 0.2607 & 3.21 & 0.0216 & 0.608 \\
BC+D
& 0.0466 & 0.2713 & \textbf{1.82} & \textbf{0.0097} & 1.06 \\
Exact-$Q$
& 0.0558 & 0.2608 & 2.78 & 0.0152 & 0.812 \\
D+Exact-$Q$
& 0.0394 & 0.2715 & 2.61 & 0.0192 & 1.40 \\
Hybrid
& 0.0430 & 0.2680 & 2.89 & 0.0350 & 7.67 \\
D+Hybrid
& \textbf{0.0388} & \textbf{0.2538} & 2.31 & 0.0165 & 2.00 \\
\bottomrule
\end{tabular}
\end{table}


\begin{figure}[H]
    \centering
    \includegraphics[
        width=\columnwidth,
        height=1.5\textheight,
        keepaspectratio
    ]{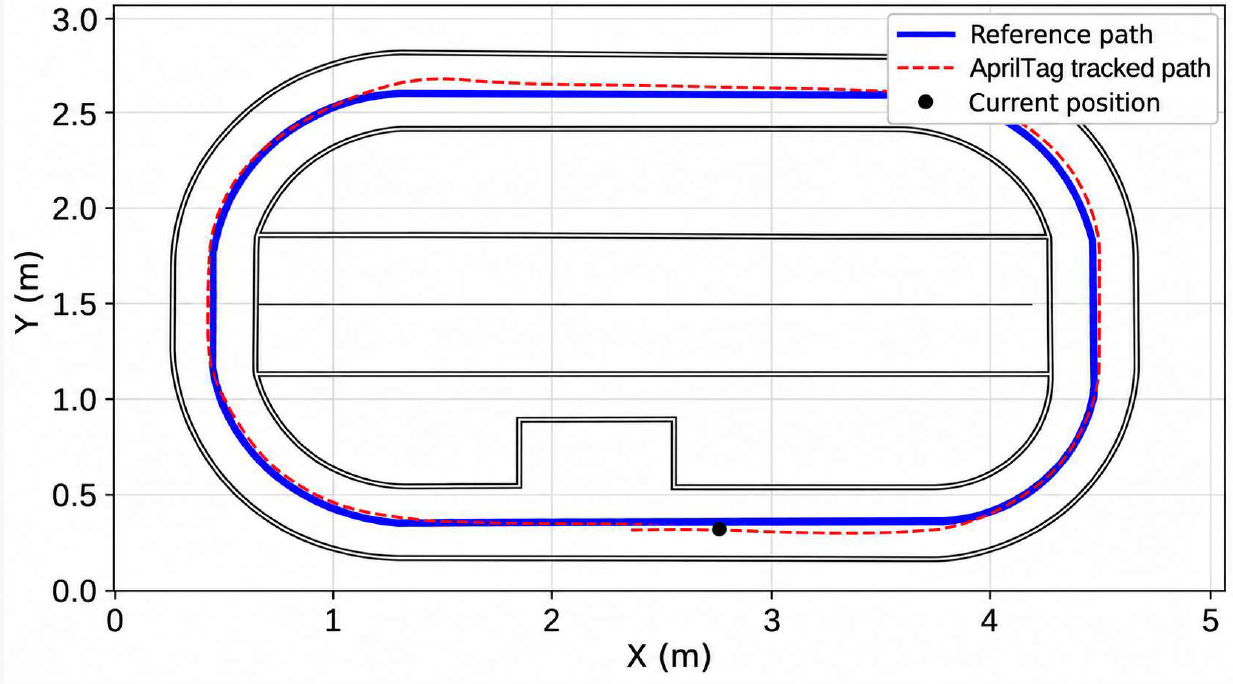}
    \caption{Representative trajectory measured by the external
    AprilTag localization system relative to the CAD/DXF reference
    path.}
    \label{fig:tracking_plot}
\end{figure}

\begin{figure}[H]
    \centering
    \includegraphics[width=\columnwidth]
    {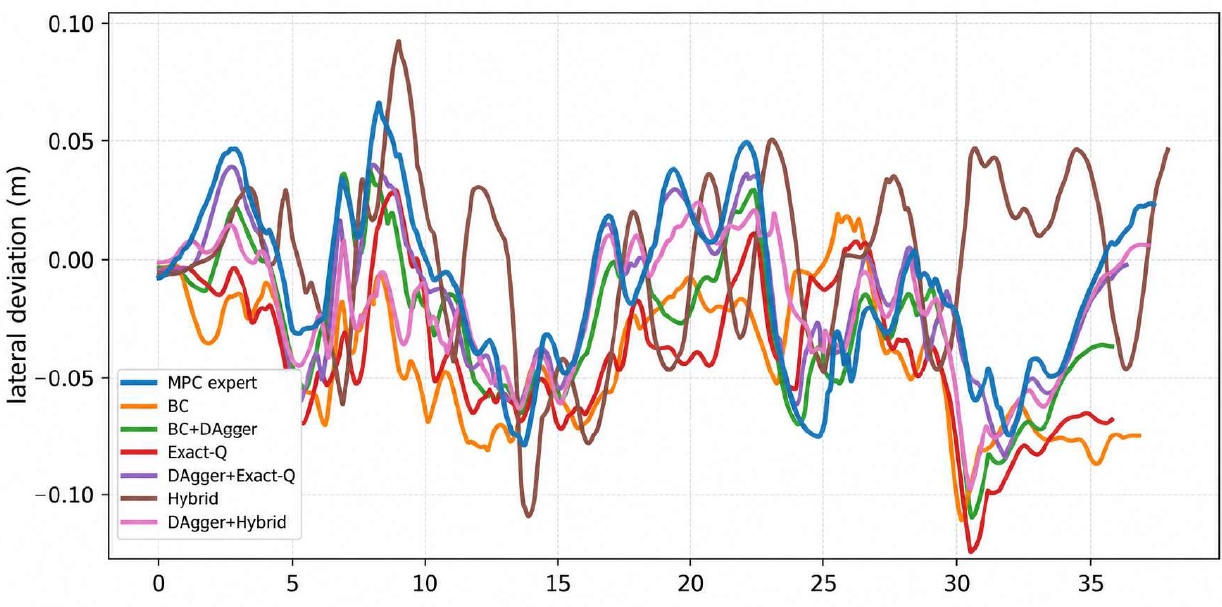}
    \caption{Lateral deviation of the MPC expert and learned policies
    during representative real-vehicle rollouts.}
    \label{fig:lateral_tracking}
\end{figure}

\begin{figure}[H]
    \centering
    \includegraphics[width=\columnwidth]
    {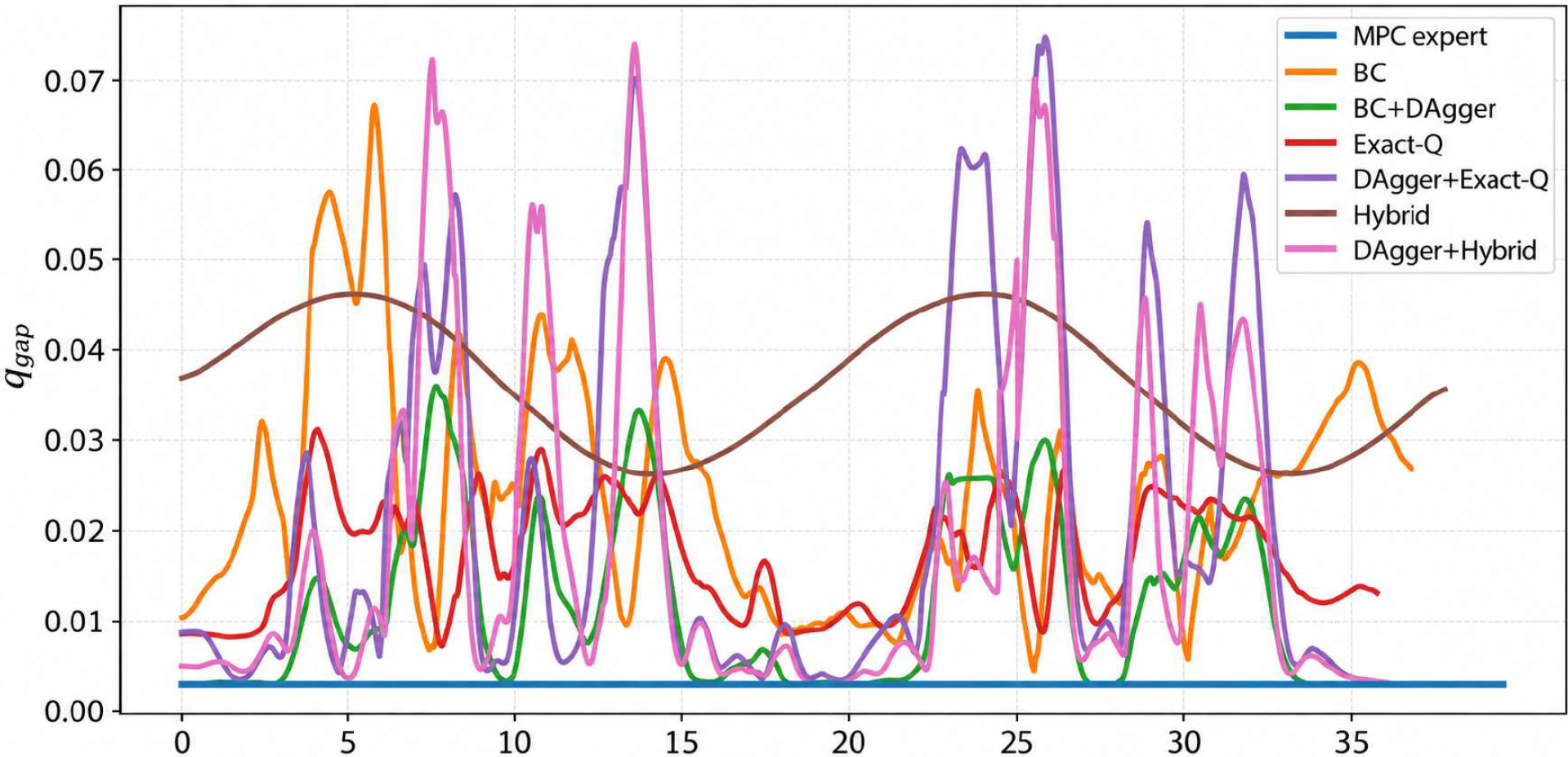}
    \caption{Exact-$Q$ gap of the learned policies relative to the MPC
    expert during representative real-vehicle rollouts.}
    \label{fig:exact_q_gap}
\end{figure}

Figures~\ref{fig:tracking_plot}--\ref{fig:exact_q_gap} and
Table~\ref{tab:real_results} provide complementary physical,
time-domain, and task-aware evaluations of the learned controllers.
Figure~\ref{fig:tracking_plot} shows that the complete
perception--control pipeline maintains the vehicle close to the
CAD/DXF reference path during real operation, while
Fig.~\ref{fig:lateral_tracking} resolves the temporal tracking behavior
that is compressed into the scalar lateral RMSE values of
Table~\ref{tab:real_results}. In contrast,
Fig.~\ref{fig:exact_q_gap} evaluates the learned steering actions
through the finite-horizon MPC performance degradation quantified by
\(
\Delta_Q
\),
and therefore measures a different notion of approximation quality
from pointwise steering error.

This distinction is consistent with
Remark~\ref{rem:q_vs_instantaneous}. The steering MAE measures direct
action agreement with the expert policy
\(
\pi_E
\),
whereas
\(
\Delta_Q
\)
measures the increase in optimal finite-horizon cost caused by fixing
the learner action and re-optimizing the remaining MPC horizon.
By Proposition~\ref{prop:q_minimizer},
\(
\Delta_Q\geq0
\)
for every feasible learner action; hence, the values in
Fig.~\ref{fig:exact_q_gap} can be interpreted directly as
task-dependent degradation relative to the MPC expert rather than as
a purely geometric action mismatch.

DAgger improves lateral tracking in all three policy families. The
lateral RMSE decreases from
\(
0.0565
\)
to
\(
0.0466~\mathrm{m}
\)
for BC, from
\(
0.0558
\)
to
\(
0.0394~\mathrm{m}
\)
for Exact-$Q$, and from
\(
0.0430
\)
to
\(
0.0388~\mathrm{m}
\)
for Hybrid. This consistent reduction supports the motivation for
certified DAgger developed above: aggregation at states actually
visited by the learner reduces the expert--learner state-distribution
mismatch created by compounding closed-loop errors, while restricting
rollouts to certified policies preserves the stability guarantees used
during data collection.

Among the learned controllers, DAgger+Hybrid achieves the lowest
lateral RMSE
\(
0.0388~\mathrm{m}
\)
and heading RMSE
\(
0.2538~\mathrm{rad}
\),
whereas BC+DAgger yields the smallest steering MAE
\(
1.82^\circ
\)
and mean Exact-$Q$ gap
\(
0.0097
\).
Moreover, introducing the Exact-$Q$ objective alone reduces the mean
\(Q\)-gap from
\(
0.0216
\)
for BC to
\(
0.0152
\)
for Exact-$Q$, which is consistent with the intended task-aware
training effect. Importantly, the policy with the smallest steering
imitation error is not the policy with the smallest closed-loop
tracking error. This provides an experimental counterpart to
Remark~\ref{rem:q_vs_instantaneous}: closer pointwise reproduction of
the MPC steering command does not necessarily imply a smaller
downstream lateral--yaw consequence.

The Hybrid results further indicate that the two training signals are
complementary rather than interchangeable. The imitation term promotes
local agreement with the expert steering map, whereas Exact-$Q$
incorporates the finite-horizon consequence of the current learner
action. Taken together, the real-vehicle results support the theoretical
motivation developed above: DAgger addresses deployment-induced
distribution shift, Exact-$Q$ provides task-aware supervision beyond
instantaneous action matching, and their combination can improve
physical closed-loop tracking without requiring exact pointwise
reproduction of the MPC steering policy.

\begin{table}[H]
\centering
\caption{Lyapunov--IQC certification of the learned neural policies.
The certificate margin is defined as
\(\mu_{\mathrm{cert}}
:=-\lambda_{\max}(\calM(\Pm,\Lam,\thetav))\).}
\label{tab:certificate_margin}
\renewcommand{\arraystretch}{1.08}
\setlength{\tabcolsep}{3.5pt}
\begin{tabular}{lccc}
\hline
\textbf{Policy}
&
\(\boldsymbol{\lambda_{\max}(\calM)}\)
&
\(\boldsymbol{\mu_{\mathrm{cert}}}\)
&
\textbf{Cert.}
\\
\hline
BC              & \(-0.0214\) & \(0.0214\) & Yes \\
Exact-\(Q\)     & \(-0.0278\) & \(0.0278\) & Yes \\
Hybrid          & \(-0.0316\) & \(0.0316\) & Yes \\
D+BC            & \(-0.0249\) & \(0.0249\) & Yes \\
D+Exact-\(Q\)   & \(-0.0307\) & \(0.0307\) & Yes \\
D+Hybrid        & \(-0.0362\) & \(0.0362\) & Yes \\
\hline
\end{tabular}
\end{table}

Finally, Table~\ref{tab:certificate_margin} reports the Lyapunov--IQC
certificate for every deployed policy. All six policies satisfy
\(\lmax(\calM)<0\), and the margin \(\mu_{\mathrm{cert}}\) increases
monotonically from BC to Hybrid and again after DAgger refinement,
with D+Hybrid attaining the largest margin (\(0.0362\)). Hence the
tracking improvements above were obtained without leaving the
certified set \(\calS_{\mathrm{cert}}\), which is precisely the
guarantee required by Lemma~\ref{lem:certified_rollouts} for the rollouts used during data
aggregation.

The physical experiments demonstrate that action imitation, task-aware
Exact-$Q$ evaluation, and closed-loop tracking provide complementary
measures of neural-MPC deployment quality.

\section{Conclusion}

This paper presented a certified exact-$Q$ imitation-learning framework
for neural approximation of MPC-based lateral control. The
path-referenced bicycle model used an MPC expert as both demonstrator
and finite-horizon value oracle by fixing each learner action as the
first input and re-optimizing the remaining horizon. The resulting
exact-$Q$ gap was combined with behavioral imitation and a
Lyapunov--IQC barrier. An LFT neural representation with sector IQCs
and a dense matrix inequality certified asymptotic stability, while
safe projection and certified DAgger kept data aggregation within the
certified policy set. Experiments on a CAD-referenced,
AprilTag-localized vehicle validated the integration of exact MPC value
imitation and certified neural closed-loop control.

\section*{Acknowledgment}
We acknowledge Ho Chi Minh City University of Technology (HCMUT), VNU-HCM for supporting this study.

\end{document}